%% file: finalcomp.tex
\documentclass[]{llncs}
\usepackage{lipsum}

\usepackage{floatflt, float}
\usepackage{url}

\usepackage{wrapfig}

\usepackage{xcolor}
\usepackage{amsmath}
\usepackage{amssymb, amsbsy}
\usepackage{array}
\usepackage{cite}
\usepackage{wrapfig}
\usepackage{multicol}
\usepackage{graphicx}
\graphicspath{ {images/} }
\usepackage{epstopdf}
\usepackage{color, import}
\usepackage{sectsty}
\usepackage{enumitem}
\usepackage{lineno,hyperref}
\usepackage[linesnumbered,vlined,ruled]{algorithm2e}
\SetAlFnt{\scriptsize}

\usepackage{setspace}

\usepackage[list=true]{subcaption}

\newcolumntype{P}[1]{>{\centering\arraybackslash}p{#1}}
\newcolumntype{M}[1]{>{\centering\arraybackslash}m{#1}}

\begin{document}

\title{Power of Finite Memory and Finite Communication Robots under Asynchronous Scheduler }
\author{Archak Das \and Avisek Sharma  \and Buddhadeb Sau}

\institute{Department of Mathematics, Jadavpur University, Kolkata, India \\
\email{\{archakdas.math.rs, aviseks.math.rs,  buddhadeb.sau\}@jadavpuruniversity.in}
}
\maketitle

\begin{abstract}
In swarm robotics, a set of robots has to perform a given task with specified internal capabilities (model) and under a given adversarial scheduler. Relation between a model $M_1$ under scheduler $S_1$, and that of a model $M_2$ under scheduler $S_2$ can be of four different types: not less powerful, more powerful, equivalent and orthogonal. In literature there are four main models of robots with lights: $\mathcal{LUMI}$, where robots have the power of observing the lights of all the robots, $\mathcal{FSTA}$ , where each robot can see only its own light, $\mathcal{FCOM}$, where each robot can observe the light of all other robots except its own and $\mathcal{OBLOT}$, where the robots do not have any light.  In this paper, we investigate the computational power of   $\mathcal{FSTA}$  and $\mathcal{FCOM}$  model under asynchronous scheduler by comparing it with other model and scheduler combinations. Our main focus is to understand and compare the power of persistent memory and explicit communication in robots under asynchronous scheduler.

\end{abstract}


\section{Introduction}
A large amount of research in the area of distributed computing has been devoted to the study of computational capabilities of autonomous mobile entities called $robots$ operating in an Euclidean plane. Each of the robots can freely move through out the space, and each are endowed with its individual co-ordinate system and operate in $Look-Compute-Move$ (LCM) cycles.   One of the main research focus has been to understand the extent of computational power of robots with minimum possible internal capabilities such as memory or communication. Some of the most notable tasks are Gathering \cite{ FlocchiniPSW05, IzumiSKIDWY12, SuzukiY99},  Pattern Formation \cite{ SuzukiY99, YamashitaS10, YamauchiUKY17}, etc.

An external factor which plays a fundamental role in determining the computational power of a swarm of mobile robots is the $activation$  $schedule$. With respect to this factor there are three different settings: $fully-synchronous$ or $FSYNCH$, $semi-synchronous$ or $SSYNCH$, and $asynchronous$ or $ASYNCH$. In $asynchronous$ setting there is no common notion of time, while $synchronous$ setting time is divided into discrete intervals called $rounds$. Among the synchronous schedulers if in each round all the robots are activated then the scheduler is said to be $fully-synchronous$. If there is no such constraint, the scheduler is said to be $semi-synchronous$.  

 In addition to external capabilities the robots can be further differentiated based on the presence and absence of internal capabilities such as memory or communication. Our main focus in this paper are  $robots$ $with$ $light$, where the lights act as a basis of memory or communication. There are four main models based on different labels of internal capabilities of memory and communication of the robots: $\mathcal{OBLOT}$, $\mathcal{LUMI}$, $\mathcal{FSTA}$, and $\mathcal{FCOM}$. In $\mathcal{OBLOT}$ model robots neither have persistent memory nor have the power of communication. In the $\mathcal{LUMI}$ model the robots are equipped with a constant-sized memory (called $light$), whose value can be set in the $Compute$ phase. The light is visible to all the robots and is persistent. Hence, the robots are able to remember and communicate a constant number of bits through the light. Now $\mathcal{FSTA}$ and $\mathcal{FCOM}$ are two sub-models of $\mathcal{LUMI}$ where in the former each robot can see only its own light and in the latter, the robot can see the lights of all the other robot except its own. Thus in $\mathcal{FSTA}$ the robots are $silent$ but has a constant amount of persistent memory, while in $\mathcal{FCOM}$ the robots are $oblivious$ but can communicate a constant number of bits. $\mathcal{X}^{Y}$ denotes model $\mathcal{X}$ under scheduler $Y$. $A>B$ indicates that model $A$ is computationally more powerful than model $B$, $A \equiv B$ denotes that they are computationally equivalent, $A \perp B$ denotes that they are computationally incomparable.
 
\paragraph{Previous Works}

 It has been shown in \cite{SuzukiY99} that within $\mathcal{OBLOT}$, robots in $FSYNCH$ are strictly more powerful than those in $SSYNCH$. In \cite{Daslumi16} it has been shown that within $\mathcal{LUMI}$, robots in $ASYNCH$ has the computational power as the robots in $SSYNCH$, and that asynchronous luminous robots are strictly more powerful than oblivious synchronous robots. 
 Recently in \cite{Flocchiniopodis}, the authors provided a complete exhaustive map of computational relationships among various models given in the tables \ref{tab:1}, \ref{tab:2}, \ref{tab:3}, but only under $FSYNCH$ and $SSYNCH$ scheduler. In addition to that they considered the question that which of the internal capabilities is more important, memory or communication. The authors examined this question through the lens of various computational results they obtained and came to the conclusion that the answer depends on the type of scheduler, i.e., communication is more powerful than persistent memory if the scheduler is fully synchronous and are incomparable under the semi-synchronous.
 In \cite{peterscomparison} the author further expanded the work of \cite{Flocchiniopodis}. The author investigated how the models relate to each other in a single robot system and  also showed that the results obtained for rigid robots with same chirality in tables \ref{tab:1} and \ref{tab:2} are the same if non-rigid robots are considered. The author also investigated the point between \textsc{Ssynch} and \textsc{Fsynch} where the strict dominance of $\mathcal{LUMI}$ over $\mathcal{FCOM}$ turns into equivalence. 
 \begin{table}[ht!]
 \footnotesize
 \begin{minipage}[b]{0.45\linewidth}\centering
     
     \begin{tabular}{|c|c|c|c|}
     \hline
            & $\mathcal{FCOM}^F$ & $\mathcal{FSTA}^F$ &  $\mathcal{OBLOT}^F$ \\
          \hline
           $\mathcal{LUMI}^F$& $\equiv$  & $>$ & $>$\\
          \hline
          $\mathcal{FCOM}^F$ & $-$ & $>$ & $>$\\
          \hline
          $\mathcal{FSTA}^F$ & $-$ & $-$ & $>$\\
          \hline
     \end{tabular}
     \caption{Relationships within \textsc{Fsynch} \cite{Flocchiniopodis} }
     \label{tab:1}
     \end{minipage}
     \hfill
 \begin{minipage}[b]{0.45\linewidth}\centering
     
     \begin{tabular}{|c|c|c|c|}
     \hline
          & $\mathcal{FCOM}^S$ & $\mathcal{FSTA}^S$ & $\mathcal{OBLOT}^S$  \\
          \hline
          $\mathcal{LUMI}^S$ & $>$ & $>$ & $>$\\
          \hline
          $\mathcal{FCOM}^S$ & $-$ & $\perp$ & $>$\\
          \hline
          $\mathcal{FSTA}^S$ & $-$ & $-$ & $>$\\
          \hline
     \end{tabular}
     \caption{Relationships within \textsc{Ssynch} \cite{Flocchiniopodis}}
     \label{tab:2}
     \end{minipage}
 \end{table}
 
 \begin{table}[ht!]
 \footnotesize
     \centering
     \begin{tabular}{|c|c|c|c|c|}
     \hline
          & $\mathcal{LUMI}^S$ & $\mathcal{FCOM}^S$ & $\mathcal{FSTA}^S$ & $\mathcal{OBLOT}^S$  \\
          \hline
          $\mathcal{LUMI}^F$ $\equiv$ $\mathcal{FCOM}^F$ & $>$ & $>$ & $>$ & $>$\\
          \hline
          $\mathcal{FSTA}^F$ & $\perp$ & $\perp$ & $>$ & $>$\\
          \hline
          $\mathcal{OBLOT}^F$ & $\perp$ & $\perp$ & $\perp$ & $>$\\
          \hline
     \end{tabular}
     \caption{Relationship between \textsc{Fsynch} and \textsc{Ssynch} \cite{Flocchiniopodis}}
     \label{tab:3}
 \end{table}

\paragraph{Our Contributions}
  In this paper, we partially extend the work of \cite{Flocchiniopodis} under $ASYNCH$ scheduler as well. We first observe that, when we consider the four models under asynchronous scheduler, some of the results follow in a straight-forward manner from the previous results established in the papers \cite{Flocchiniopodis,Daslumi16, Conrend}. For, some of the remaining cases which requires a deeper investigation we give non-trivial proofs. For example we prove that $\mathcal{FSTA}^{A}$ $\perp$ $\mathcal{OBLOT}^{F}$, and we prove that $\mathcal{FCOM}^{A}$ $\perp$ $\mathcal{OBLOT}^{F}$, which is a clear refinement of the recently established result that $\mathcal{LUMI}^{A}$ $\perp$ $\mathcal{OBLOT}^{F}$ \cite{Flocchiniopodis}. To prove $\mathcal{FSTA}^{A}$ $\perp$ $\mathcal{OBLOT}^{F}$ we introduce a new problem $Oscillating$ $Configurations$ while to prove $\mathcal{FCOM}^{A}$ $\perp$ $\mathcal{OBLOT}^{F}$, we provide a new asynchronous algorithm for the problem $-IL$ introduced in  \cite{Flocchiniopodis}. By designing an algorithm for $-IL$ in  $\mathcal{FCOM}^{A}$ we also prove that $\mathcal{FCOM}^{A}$ $\perp$ $\mathcal{FSTA}^{F}$, which is in contrast to the fact that $\mathcal{FCOM}^{F}$ $>$ $\mathcal{FSTA}^{A}$, which is easily deduced from some previous results. We introduce a new problem $Independent$ $Oscillating$ $Problem$ and prove that $\mathcal{FCOM}^{S} \perp \mathcal{FSTA}^{A}$ which in turn implies that $\mathcal{FCOM}^{A} \perp \mathcal{FSTA}^{A}$. In \cite{Flocchiniopodis}, the question was raised that whether ``\textit{ it was better to remember or to communicate?}". Thus we have shown here that in the case of asynchronous scheduler, finite memory model is incomparable to finite communication model.

\section{Model and Technical Preliminaries}

The model and setting in this paper is same as that of \cite{Flocchiniopodis}. The System considered in this paper consists of a team $R$ of computational entities called $robots$ moving and operating in the Euclidean plane $\mathbb{R}^2$. The robots are viewed as points in the Euclidean plane and has its own local co-ordinate system which may not agree with other robots. It always perceives itself at the origin. The robots are $identical$ (indistinguishable) and $autonomous$ (lacks a central control).At any point in time, a robot is either $active$ or $inactive$. When active, a robot executes a $Look-Compute-Move$ $(LCM)$ cycle performing the following three operations:
\begin{enumerate}
    \item $Look$: The robot activates its sensors to obtain a snapshot of the positions occupied by robots with respect to its own co-ordinate system.
    \item $Compute$: The robot executes its algorithm using the snapshot as input. The result of the computation is a destination point.
    \item $Move$: The robot moves to the computed destination, i.e., we assume a $rigid$ $mobility$. If the destination is the current location, the robot stays still.
\end{enumerate}

When inactive, a robot is idle. All robots are initially idle. The amount of time to complete a cycle is assumed to be finite, and the $Look$ operation is assumed to be instantaneous. We assume that the robots agree on the same circular orientation of the plane, i.e., there is $chirality$.

In the $\mathcal{OBLOT}$ model, the robots are $silent$ and $oblivious$. By $silent$ we mean that the robots have no explicit means of communication, and by $oblivious$ we mean that, at the start of a cycle, a robot has no memory of observations and computations performed in previous cycles.

In the $\mathcal{LUMI}$ model, each robot $r$ is equipped with a persistent visible state variable $Light[r]$, called $light$, whose values are taken from a finite set $C$ of states called colors. The colors of the lights can be set in each cycle by $r$ at the end of its $Compute$ operation. A light is persistent from one computational cycle to the next and can be seen by all the robots. The robot is otherwise oblivious forgetting all other information from previous cycles. In $\mathcal{LUMI}$, the $Look$ operation produces the set of pairs $(position, color)$ of the other robots. It can be clearly understood that the lights simultaneously provide persistent memory and direct means of communication. 

$\mathcal{FSTA}$ is a sub-model of $\mathcal{LUMI}$ having only the property of persistent memory. Here the robots can only see the color of its own light i.e. the light is an $internal$ $light$. As a result the robots are $silent$ as in $\mathcal{OBLOT}$ but $finite-state$. 

$\mathcal{FCOM}$ is a sub-model of $\mathcal{LUMI}$ having only the property of external communication. Here the lights of the robot are $external$ $light$ i.e. a robot cannot see the color of its own light but can see the color of the lights of the other robots. As a result the robots are $oblivious$ as in $\mathcal{OBLOT}$ but are $finite-communication$.

With respect to the activation schedule of the robots, there are three kind of schedulers. In the $asynchronous$ (ASYNCH) scheduler, there is no common notion of time, each robot is activated independently of others, the duration of each phase is finite but unpredictable and might be different in different cycles. In the $semi-synchronous$ (SSYNCH) setting, time is divided into discrete intervals, called $rounds$, in each round some robots are activated simultaneously, and perform their $LCM$ cycles in perfect synchronization. The $fully-synchronous$ (FSYNCH) setting, is same as SSYNCH except the fact that each robot is activated in every round. In all of these settings, the selection of which robots are activated at a round is made by an $adversarial$ $scheduler$, which must be $fair$, i.e., every robot must be activated infinitely often.

The Computational Relationships we now define below are exactly the same as \cite{Flocchiniopodis}. For the sake of completion of the paper we  re-write it.

Let $\mathcal{M}$ $=$ $\{ \mathcal{LUMI, FSTA, FCOM, OBLOT} \}$ be the set of models, and $\mathcal{S}$ $=$ $\{ FSYNCH, SSYNCH, ASYNCH \}$ the set of schedulers. $\mathcal{R}$ denotes the set of all team of robots satisfying core assumptions (i.e, they are identical, autonomous and operate in $LCM$ cycles) and $R$ $\in$ $\mathcal{R}$, a team of robots having identical capabilities. By $\mathcal{R}_n$ $\in$ $\mathcal{R}$ we denote the set of all teams of size $n$.

Given a model $M$ $\in$ $\mathcal{M}$, a scheduler $S$ $\in$ $\mathcal{S}$, and a team of robots $R$ $\in$ $\mathcal{R}$, let $Task(M,S;R)$ denote the set of problems solvable by $R$ in $M$ under adversarial scheduler $S$.

Let $M_1$, $M_2$ $\in$ $\mathcal{M}$ and $S_1$, $S_2$ $\in$ $\mathcal{S}$. We define the relationships between model $M_1$ under scheduler $S_1$ and model $M_2$ under scheduler $S_2$:

\begin{itemize}
    \item $computationally$ $not$ $less$ $powerful$ ($M_1^{S_1}$ $\geq$ $M_2^{S_2}$), if $\forall$ $R \in \mathcal{R}$ we have $Task(M_1, S_1;R)$ $\supseteq$ $Task(M_2, S_2;R)$;

    \item $computationally$ $more$ $powerful$ ($M_1^{S_1}$ $>$ $M_2^{S_2}$), if $M_1^{S_1}$ $\geq$ $M_2^{S_2}$ and $\exists R \in \mathcal{R}$ such that $Task(M_1, S_1;R)$ $\setminus$ $Task(M_2, S_2;R)$ $\neq$ $\emptyset$;
    
    \item $computationally$ $equivalent$ ($M_1^{S_1}$ $\equiv$ $M_2^{S_2}$), if $M_1^{S_1}$ $\geq$ $M_2^{S_2}$ and $M_1^{S_1}$ $\leq$ $M_2^{S_2}$;

    \item $computationally$ $orthogonal$ $or$ $incomparable$, ($M_1^{S_1}$ $\perp$ $M_2^{S_2}$), if $\exists R_1, R_2 \in \mathcal{R}$ such that $Task(M_1, S_1;R_1)$ $\setminus$ $Task(M_2, S_2;R_1)$ $\neq$ $\emptyset$ and $Task(M_2, S_2;R_2)$ $\setminus$ $Task(M_1, S_1;R_2)$ $\neq$ $\emptyset$.
\end{itemize}

For simplicity of notation, for a model $M$ $\in$ $\mathcal{M}$, let $M^{F}$ and $M^{S}$ denote $M^{Fsynch}$ and $M^{Ssynch}$, respectively; and let $M^{F}(R)$ and $M^{S}(R)$ denote the sets $Task(M,FSYNCH;R)$ and $Task(M,SSYNCH;R)$, respectively.

Trivially, for any $M$ $\in$ $\mathcal{M}$, $M^{F} \geq M^{S} \geq M^{A}$; also for any $S$ $\in$ $\mathcal{S}$, $\mathcal{LUMI}^{S} \geq \mathcal{FSTA}^{S} \geq \mathcal{OBLOT}^{S}$ and $\mathcal{LUMI}^{S} \geq \mathcal{FCOM}^{S} \geq \mathcal{OBLOT}^{S}$.

\section{Obvious Deductions from previous algorithms}\label{3.2}
A number of results has been established previously in the papers \cite{Flocchiniopodis,Daslumi16, Conrend} but a most of them are considering $FSYNCH$ and $SSYNCH$ schedulers only. When we change the scheduler to $asynchronous$ many new results follow easily :
\begin{itemize}
    \item (D1) $\mathcal{FSTA}^{F}$ $>$ $\mathcal{OBLOT}^{A}$ 
    \item (D2) $\mathcal{FCOM}^{F}$ $>$ $\mathcal{OBLOT}^{A}$
    \item (D3) $\mathcal{LUMI}^{F}$ $>$ $\mathcal{OBLOT}^{A}$
    \item (D4) $\mathcal{LUMI}^{S}$ $>$ $\mathcal{OBLOT}^{A}$
    \item (D5) $\mathcal{FCOM}^{F}$ $>$ $\mathcal{FSTA}^{A}$
   
    \item (D6) $\mathcal{FSTA}^{F}$ $>$ $\mathcal{FSTA}^{A}$
    \item (D7) $\mathcal{LUMI}^{A}$ $>$ $\mathcal{FSTA}^{A}$
    \item (D8) $\mathcal{LUMI}^{S}$ $>$ $\mathcal{FSTA}^{A}$
    \item (D9) $\mathcal{LUMI}^{F}$ $>$ $\mathcal{FSTA}^{A}$
    
    \item (D10) $\mathcal{FCOM}^{F}$ $>$ $\mathcal{FCOM}^{A}$
    \item (D11) $\mathcal{LUMI}^{A}$ $>$ $\mathcal{FCOM}^{A}$
    \item (D12) $\mathcal{LUMI}^{S}$ $>$ $\mathcal{FCOM}^{A}$
    \item (D13) $\mathcal{LUMI}^{F}$ $>$ $\mathcal{FCOM}^{A}$
    \item (D14) $\mathcal{FCOM}^{F}$ $>$ $\mathcal{LUMI}^{A}$
    \item (D15) $\mathcal{LUMI}^{F}$ $>$ $\mathcal{LUMI}^{A}$
    \item (D16) $\mathcal{FSTA}^{S}$ $<$ $\mathcal{LUMI}^{A}$
    \item (D17) $\mathcal{FCOM}^{S}$ $<$ $\mathcal{LUMI}^{A}$
\end{itemize}

\begin{proof}
\begin{enumerate}
    \item (D1) and (D2) follows from the fact that $\mathcal{FSTA}^{F}$ $>$ $\mathcal{OBLOT}^{F}$ and $\mathcal{FCOM}^{F}$ $>$ $\mathcal{OBLOT}^{F}$ \cite{Flocchiniopodis} and $\mathcal{OBLOT}^{F}$ $>$ $\mathcal{OBLOT}^{A}$ \cite{SuzukiY99}
    \item (D3) and (D4) follows from the fact that $\mathcal{LUMI}^{A}$ $>$ $\mathcal{OBLOT}^{A}$ \cite{Daslumi16}
    \item (D5) folows from the fact that $\mathcal{FCOM}^{F}$ $>$ $\mathcal{FSTA}^{S}$ \cite{Flocchiniopodis}.  (D6) follows from the fact that $\mathcal{FSTA}^{F}$ $>$ $\mathcal{FSTA}^{S}$ \cite{Flocchiniopodis}.
    \item (D7) follows from $\mathcal{LUMI}^{S}$ $\equiv$ $\mathcal{LUMI}^{A}$ and $\mathcal{LUMI}^{S}$ $>$ $\mathcal{FSTA}^{S}$  \cite{Flocchiniopodis}. (D8) and (D9) follows from (D7).
    \item (D10) follows from the fact that $\mathcal{FCOM}^{F}$ $>$ $\mathcal{FCOM}^{S}$ \cite{Flocchiniopodis}.
    \item (D12) follows from $\mathcal{LUMI}^{S}$ $>$ $\mathcal{FCOM}^{S}$ \cite{Flocchiniopodis}. (D11) follows  from $\mathcal{LUMI}^{S}$ $\equiv$ $\mathcal{LUMI}^{A}$ and (D12). (D13) follows from (D11) and (D12).
    \item (D14), (D15) follows from $\mathcal{FCOM}^{F}$ $\equiv$ $\mathcal{LUMI}^{F}$ and $\mathcal{LUMI}^{F}$ $>$ $\mathcal{LUMI}^{S}$ \cite{Flocchiniopodis}.
    \item (D16), (D17) follows from the fact that $\mathcal{LUMI}^{S}$ $\equiv$ $\mathcal{LUMI}^{A}$ and $\mathcal{LUMI}^{S}$ $>$ $\mathcal{FSTA}^{S}$, $\mathcal{LUMI}^{S}$ $>$ $\mathcal{FCOM}^{S}$ \cite{Flocchiniopodis}.
\end{enumerate}
\qed
\end{proof}

There are still some non-trivial relationships whose nature still remained unresolved and which cannot be deduced from results of the previous papers. In the next section, in sub-section \ref{4.1} we prove that $\mathcal{FSTA}^{A}$ is orthogonal to $\mathcal{OBLOT}^{F}$, and strictly powerful than $\mathcal{OBLOT}^{A}$. In sub-section \ref{4.2} we prove that $\mathcal{FCOM}^{A}$ is orthogonal to $\mathcal{OBLOT}^{F}$, $\mathcal{FSTA}^{F}$ and $\mathcal{FSTA}^{S}$.

\section{Computational Power of $\mathcal{FSTA}$ and $\mathcal{FCOM}$ in Asynch}

\subsection{Problem Oscillating Configurations}\label{4.1}

We first investigate the computational relationship between $\mathcal{OBLOT}^{F}$ and $\mathcal{FSTA}^{A}$. It has already been shown in  \cite{Flocchiniopodis} that the $SRO$ (Shrinking Rotation) problem can be solved in $\mathcal{OBLOT}^{F}$ but not in $\mathcal{FSTA}^{S}$. Hence $SRO$ problem cannot be solved in $\mathcal{FSTA}^{A}$. Now we show that, there exists a problem which can be solved in $\mathcal{FSTA}^{A}$ but not in $\mathcal{OBLOT}^{F}$. We now define the problem:

\begin{definition}
Problem $OC$ (Oscillating Configurations): Four robots $x,y,z$ and $t$ are placed initially in Configuration-I represented by Figure (I). The problem $OC$ requires the robots to move to continuously alternate the configuration of the robots in the sequence I-II-III-II-I-II-...., i.e.\, the robots start at Configuration-I, form Configuration-II, then forms Configuration-III, then again forms Configuration-II,... and the process is continuous.
\end{definition}

\begin{figure}[htb!]
\centering
\fontsize{8pt}{8pt}\selectfont
\def\svgwidth{0.6\textwidth}
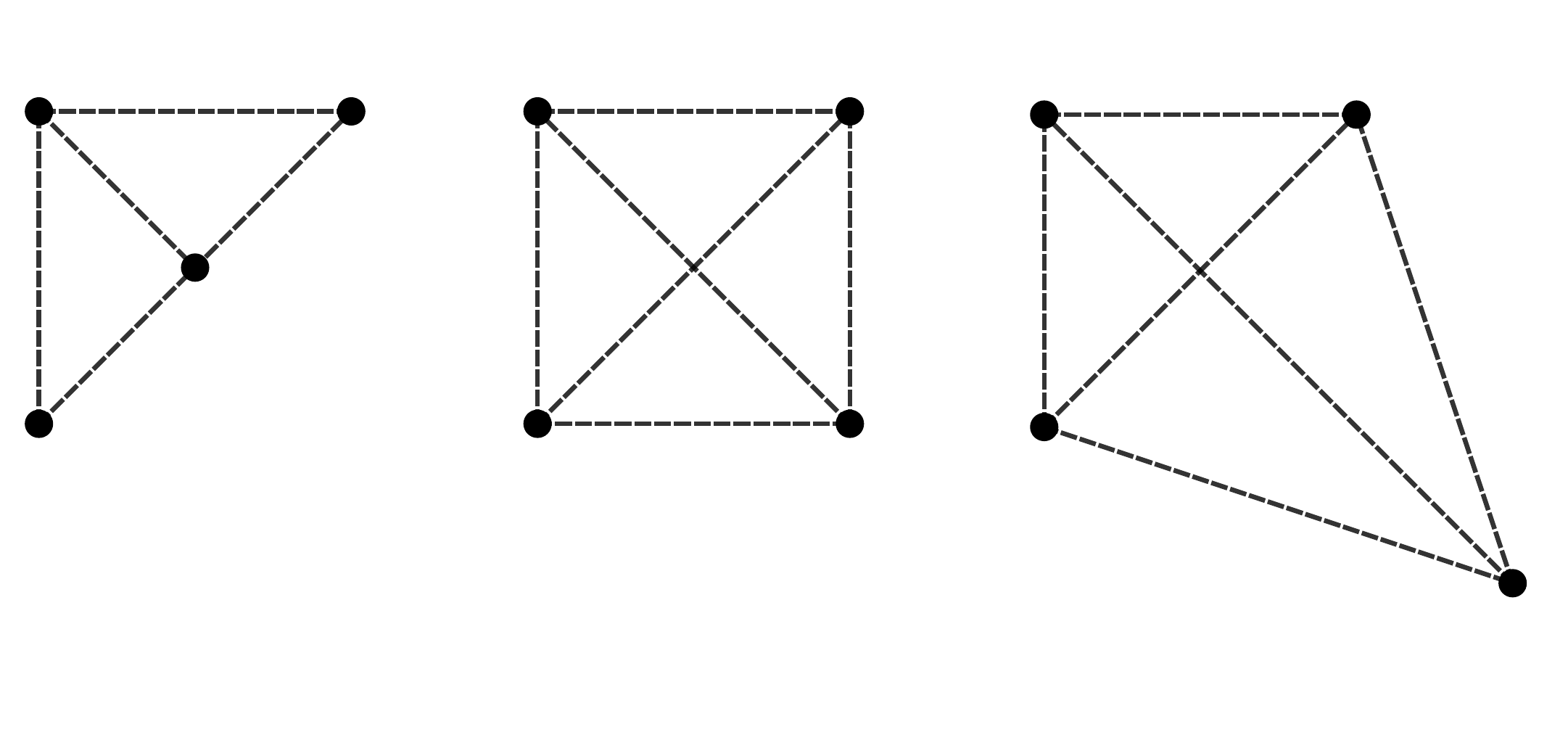
\caption{}
\label{communication_wheel_fig}
\end{figure}

\begin{algorithm}[H]
\If{r.light $=$ NIL }
{
\eIf{visible configuration is not I}
     {do nothing}
     {\eIf{r.position $\neq$ $C$}
          {do nothing}
          {r.light $\leftarrow$ RED\\
           Move to $C'$}
     }

}

\If{r.light $=$ RED}
    {\If{location is $C'$} 
         {r.light $\leftarrow$ BLUE\\
           Move to $C''$}
    }

\If{r.light $=$ BLUE}
    {\If{location is $C''$}
          {Move to $C'$}

     \If{location is $C'$}
          {r.light $\leftarrow$ NIL\\
           Move to $C$}      
    }

 \caption{Algorithm $AlgOC$ for $OC$ in $\mathcal{FSTA}^{A}$ executed by each robot $r$, initially r.light $\leftarrow$ $NIL$}
\end{algorithm}

We claim that the problem can be solved in $\mathcal{FSTA}^{A}$ but cannot be solved in $\mathcal{OBLOT}^{F}$. We prove our claim by providing an algorithm $AlgOC$ to solve $OC$ in $\mathcal{FSTA}^{A}$ and then proving that $OC$ cannot be solved in $\mathcal{OBLOT}^{F}$.

\subsubsection{Description of the Algorithm $AlgOC$}

The internal light used by the robots can have three colors: ${NIL, RED, BLUE}$. Initially all the robots start from Configuration-I and have their internal lights set to $NIL$. A robot $r$ when activated takes a snapshot of its surrounding in the $Look$ phase and tries to understand the current configuration. The robot $r$ then performs the following actions based on the dual information about its position in the current configuration and the color of its internal light:
\begin{enumerate}
    \item If the color of the internal light of $r$ is $NIL$, and the visible configuration is not I, then it does nothing in that phase,i.e., there is no change of internal light or movement on the part of $r$ in that phase. If the configuration is $I$, but the location of $r$ is not at $C$ then also it does nothing in that phase. But if the location of $r$ is at $C$, then it changes the color of its internal light to $RED$ and moves to the point $C'$.
    
    \item If the color of the internal light of $r$ is $RED$, the visible configuration is II and it is at $C'$, then $r$ changes the color of its internal light to $BLUE$, and moves to the point $C''$.
    
    \item If the color of the internal light of $r$ is $BLUE$, then if the location of $r$ is at $C''$ then $r$ moves to the point $C'$, and if the location of $r$ is at $C'$ then $r$ changes the color of its internal light to $NIL$ and moves to the point $C$.
\end{enumerate}

\subsubsection{Correctness of the Algorithm AlgOC}

\begin{lemma}\label{lma1}
$\forall R \in \mathcal{R}_4, OC \in FSTA^{A}(R)$.
\end{lemma}

\begin{proof}
A robot's internal light can take three possible colors: { $ NIL, RED , BLUE $}. Initially each robot has its initialized to $NIL$. In our algorithm, we ensure that the robot placed at position $C$ in the initial configuration, i.e\ $z$, is the only robot that makes any movement throughout the execution. We denote the position where robot $z$ must move to make the configuration II and III as $C'$ and $C''$ respectively. Initially  a robot wakes up to find its colour $NIL$, and if it is at $C$ it can clearly distinguish its position from other robots, since it is the only robot in the whole configuration which is equidistant from the other three robots . So, at that point the robot changes its colour to $RED$ and move to position $C'$ such that the resulting configuration is II. At this point $z$ is the only robot with colour $RED$. When again $z$ is activated by the adversary it wakes to find its internal light having colour $RED$ and configuration II, so it changes the color of its internal light to $BLUE$ and moves to the position $C''$ such that the resulting configuration becomes III. Next again when $z$ is activated it finds it internal light having the colour $BLUE$, and visible configuration to be III. It then  moves to the position $C'$ forming configuration II. Now when again $z$ is activated it finds the visible configuration to be configuration II and internal light $BLUE$, it then changes the color of its internal light to $NIL$ and moves to the point $C$ forming configuration $I$. When $z$ is activated again, in this position, it finds its internal light having colour $NIL$ and configuration I, then $r$ repeats its moves and the algorithm continues like this.

Note that, after the execution of the algorithm has started, if any robot other than $z$ is activated, they do not do anything, as a robot with light $NIL$ makes a movement only when the visible configuration is I and it is in $C$, and according to our algorithm only robot $z$ satisfy this criteria, throughout the execution of the algorithm.\qed
\end{proof}

\begin{lemma}\label{lma2}
$\exists R \in \mathcal{R}_4$,  $OC$ $\notin$ $OBLOT^{F}(R)$.
\end{lemma}

\begin{proof}
Let there exists an algorithm $A$ which solves problem $OC$ in $\mathcal{OBLOT}^{F}$.Now let us consider two different scenarios:\\
\textbf{Case-1}: The robots have just changed their configuration to II from I, and the configuration is currently II.\\
\textbf{Case-2}: The robots have just changed their configuration to II from III, and the configuration is currently II.

In Case-1 the next configuration that the robots must form is III, and in Case-2 the next configuration robots have to form is I. But in both the cases, the view of the robots are same. And as in $\mathcal{OBLOT}^{F}$ setting the robot has to decide its next destination based on its present view as the robots neither have any internal memory nor they have the power of external communication. As a result, the robots cannot distinguish between Case-1 and Case-2, hence they will perform the same move in both the cases. Hence, problem $OC$ cannot be solved in $\mathcal{OBLOT}^{F}$.\qed

\end{proof}

\begin{theorem}\label{thm1}
$\mathcal{OBLOT}^{F}$ $\perp$ $\mathcal{FSTA}^{A}$
\end{theorem}

\begin{proof}
Problem $OC$ can be solved in $\mathcal{FSTA}^{A}$, but cannot be solved in $\mathcal{OBLOT}^{F}$, and in \cite{Flocchiniopodis} it has been proved that the problem $SRO$ or Shrinking Rotation can be solved in $\mathcal{OBLOT}^{F}$ but cannot be solved in $\mathcal{FSTA}^{S}$, hence not in $\mathcal{FSTA}^{A}$. Hence our theorem is proved. \qed
\end{proof}

\begin{theorem}
$\mathcal{OBLOT}^{A}$ $<$ $\mathcal{FSTA}^{A}$
\end{theorem}

\begin{proof}
Problem $OC$ cannot be solved in $\mathcal{OBLOT}^{F}$, hence obviously not in $\mathcal{OBLOT}^{A}$. But problem $OC$ as proved in Theorem \ref{lma1} can be solved in $\mathcal{FSTA}^{A}$. Also, trivially $\mathcal{OBLOT}^{A}$ $\leq$ $\mathcal{FSTA}^{A}$.  Hence, the theorem follows.\qed
\end{proof}

\subsection{An Asynchronous Algorithm in $\mathcal{FCOM}$ for Problem $-IL$}\label{4.2}

We consider the problem $-IL$ defined in \cite{Flocchiniopodis}.

\begin{figure}[htb!]
\centering
\fontsize{8pt}{8pt}\selectfont
\def\svgwidth{0.6\textwidth}
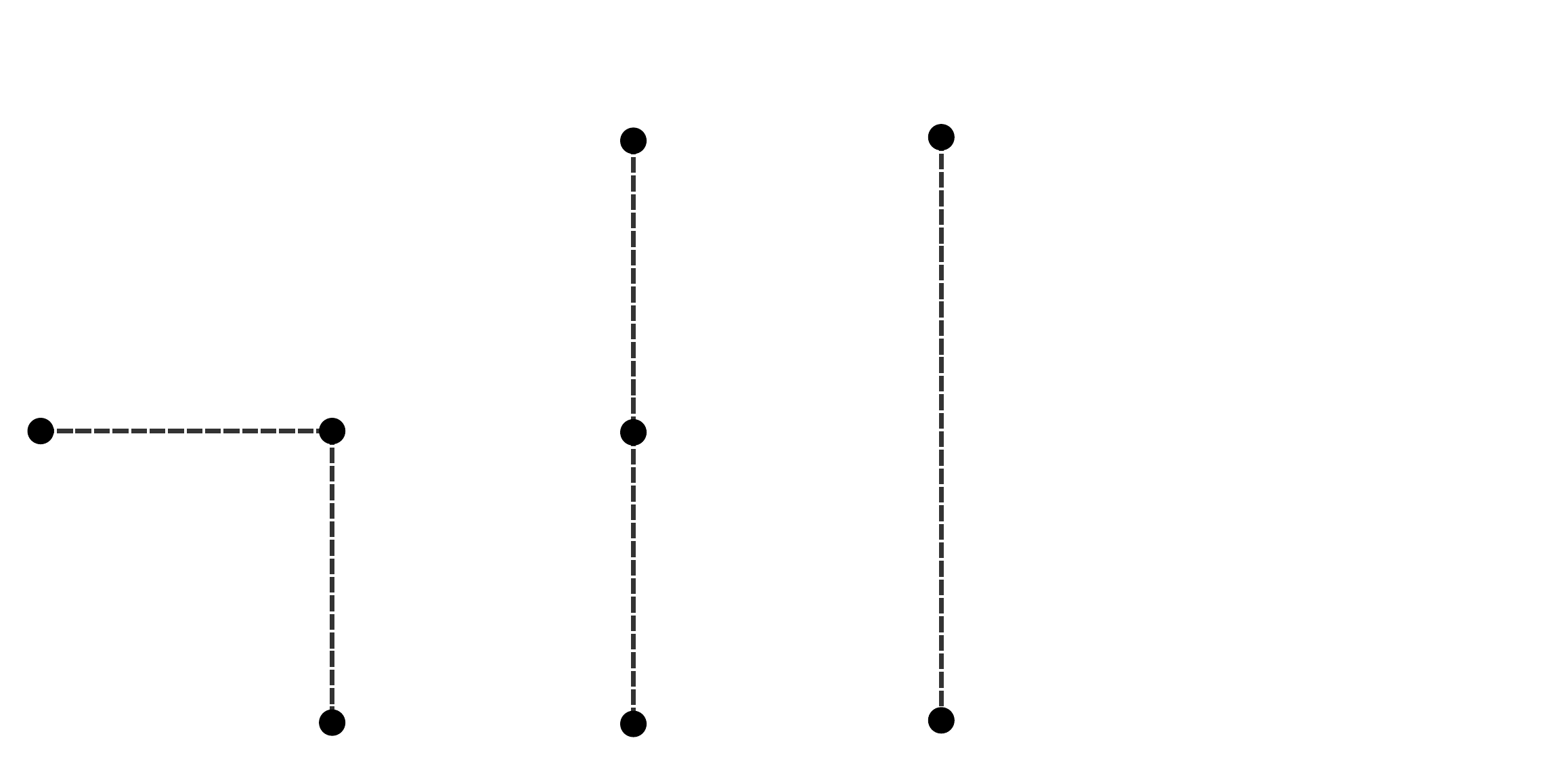
\caption{}
\label{communication_wheel_fig}
\end{figure}

\begin{definition}
Three robots $a$, $b$, and $c$, starting from the initial configuration shown in Figure 2(a), must form first the pattern of Figure 2(b) and then move to form the pattern of Figure 2(c).
\end{definition}

It has been already proved in \cite{Flocchiniopodis} that $-IL$ cannot be solved in $\mathcal{FSTA}^{F}$. In this paper, we provide an algorithm for $-IL$ in $\mathcal{FCOM}^{A}$.
We re-name Figure-2(a),2(b),2(c) in our Figure-2 as Configuration I,II and III respectively. Also we call the point where robot $a$ must move from configuration-I to form configuration-II to be $P_1$, and the point where the robot $b$ must move from configuration-II to form configuration-III to be $P_2$.

\begin{algorithm}[H]

\eIf{there exists a robot with light F}
     {do not do anything}
       {\eIf{there exists a robot with light M}
              {\eIf{Visible Configuration is not II}
                    {do not do anything}
                     {\eIf{position is in the line segment joining the other two robots}
                           {r.light $\leftarrow$ F\\
                            Move to the $P_2$}
                            {do nothing}
                     }
              }
                {\eIf{Visible Configuration is II}
                      {do nothing}
                        {\eIf{I can form II with least clockwise rotation}
                              {r.light $\leftarrow$ M\\
                                 Move to $P_1$}
                                {do nothing}
                        }
                }
       }

 \caption{Algorithm $COMIL$ for $-IL$ in $\mathcal{FCOM}^{A}$ executed by each robot $r$, initially r.light $\leftarrow$ $NIL$}
\end{algorithm}

\subsubsection{Description of Algorithm COMIL}

Here the robots are equipped with an external light whose color can only be seen by other robots. The external light can take one of the following three colors: ${NIL, M, F}$. Initially, all robots start at Configuration-I and have their external light set to $NIL$. A robot $r$ when activated takes a snapshot of its surrounding in the $Look$ phase. The snapshot  provides $r$ with the location of the robots according to its own local co-ordinate system and the current color of the light of the robots except itself. The robot $r$ then performs the following actions based on the dual information about its position in the current configuration and the colors of the external lights of the robots except itself:

\begin{enumerate}
    \item If $r$ observes any robot or robots in the snapshot obtained having external light $F$ it does nothing, i.e, it neither changes the color of its external light nor it executes any movement, irrespective of the present configuration and any other light attained in the snapshot.
    \item If $r$ does not observe any robot having external light $F$, but observes a robot having external light $M$, then there can be two cases: either the visible configuration is Configuration-II, or the visible configuration is not Configuration-II. If the configuration is not Configuration-II $r$ does not do anything. Further if the current configuration is Configuration-II, there can be two cases. If $r$ does not lie in the line segment joining the other two robots it does nothing. Otherwise if $r$ lies in the line segment joining the location of the other two robots, then $r$ changes the color of the external light to $F$ and moves to the point $P_2$.
    \item If $r$ does not observe any visible light, then there can be two cases, either the observed configuration is Configuration I or the observed configuration is Configuration II, i.e. $r$ is $a$. If the observed configuration is Configuration I, then if $r$ is $a$ ( the robot which can form Configuration II with least clockwise angular movement), then $r$ changes the color of the external light to $M$ and moves to the point $P_1$. Otherwise if the observed configuration is Configuration II, $r$ does nothing.
\end{enumerate}

\subsubsection{Correctness of the Algorithm COMIL}

\begin{lemma}\label{lma3}
$\forall R$ $\in$ $\mathcal{R}_3$, $-IL$ $\in$ $FCOM^{A}(R)$.
\end{lemma}

\begin{proof}
Initially all robots have colour NIL and are arranged in Configuration-I. Here the robot $a$ can uniquely identify its position, as it is the only robot whose ninety degree clockwise rotation around the robot $b$ results in Configuration-II. According to our algorithm, at this position only robot $a$ is allowed to move, $a$ changes its colour to $M$ and moves to position $P_1$ forming Configuration-II. While $a$ is moving towards the point $P_1$, if $b$ or $c$ is activated, their snapshot in the $Look$ returns a configuration where there is one robot with light set to $M$, and the configuration is not yet Configuration II. Hence at this point neither $b$, nor $c$ makes any movement. When $a$ is at the point $P_1$, if $a$ or $c$ is activated neither of them changes their color or make any movement. If $a$ is activated $a$ observes the color of the other two robots to be set to $NIL$ and the visible configuration is Configuration II. Hence $a$ does not do anything. If $c$ is activated it observes one robot ($a$) having external light set to $M$ and the visible configuration to be Configuration II, but $c$ does not lie in the line segment joining the other two robots, i.e., $a$ and $b$. Hence in this case $c$ also does not do anything. If $b$ is activated, it observes one robot with its external light set to $M$, visible configuration to be Configuration II, and also it lies on the line segment joining the other two robots. As the external light of $a$ is still set to $M$, $b$ clearly identifies the position $P_2$ where it must move to form Configuration III. Now it turns its light to $F$ and moves to $P_2$ which results in Configuration III. Now while $b$ is moving towards the $P_2$, if either $a$ or $c$ is activated they observe a robot with its external light set to $F$ and hence does not do anything. Once $b$ reaches $P_2$ the resulting configuration is Configuration III. If again at this point if $a$ or $c$ is activated they observe a robot with its external light set to $F$, and hence again, neither of them do anything. If $b$ is activated, $b$ observes a robot having its external light set to $M$ and the configuration is not Configuration II. Hence $b$ does not make any movement. As a result the robots do not make any further movement and remain at Configuration III permanently according to our algorithm and the problem is solved. It must be noted that as we have assumed rigid movement of the robots, so the robots $a$ and $b$ do not stop in between their movement. Hence our algorithm solves problem $-IL$.\qed
\end{proof}

\begin{theorem}\label{thm3}
$\mathcal{FCOM}^{A}$ $\perp$ $\mathcal{FSTA}^{F}$
\end{theorem}

\begin{proof}
It has been proved earlier in \cite{Flocchiniopodis} that Problem $-IL$ cannot be solved in $\mathcal{FSTA}^{F}$, and we have proved in Lemma \ref{lma3} that the Algorithm $COMIL$ solves $-IL$ in $\mathcal{FCOM}^{A}$. Also in \cite{Flocchiniopodis}, it has been proved that there exists a problem, i.e., Problem $CGE$ or Center of Gravity Expansion which can be solved in $\mathcal{FSTA}^{F}$ but cannot be solved in $\mathcal{LUMI}^{S}$. 

Now, $\mathcal{LUMI}^{S}$  $\equiv$ $\mathcal{LUMI}^{A}$ $>$ $\mathcal{FCOM}^{A}$ (Section 3), and hence Problem $CGE$ cannot be solved in $\mathcal{FCOM}^{A}$ as well. Hence, the theorem follows.\qed
\end{proof}

\begin{theorem}

$\mathcal{FCOM}^{A}$ $\perp$ $\mathcal{FSTA}^{S}$

\end{theorem}

\begin{proof}
Problem $-IL$ can be solved in $\mathcal{FCOM}^{A}$, but cannot be solved in $FSTA^{F}$ and as we know that, $\mathcal{FSTA}^{S}$ $\leq$ $\mathcal{FSTA}^{F}$, hence $-IL$ cannot be solved in $\mathcal{FSTA}^{S}$. Also in \cite{Flocchiniopodis} it has been shown that the Problem $TAR(d)$ or Triangle Rotation can be solved in $\mathcal{FSTA}^{S}$, but not in $\mathcal{FCOM}^{S}$. Now as $\mathcal{FCOM}^{A}$ $\leq$ $\mathcal{FCOM}^{S}$, so Problem $TAR(d)$ cannot be solved in $\mathcal{FCOM}^{A}$. Hence the theorem follows.\qed
\end{proof}

\begin{theorem}\label{thm5}

$\mathcal{FCOM}^{A}$ $\perp$ $\mathcal{OBLOT}^{F}$

\end{theorem}

\begin{proof}
We know that $\mathcal{OBLOT}^{F}$ $<$ $\mathcal{FSTA}^{F}$ \cite{Flocchiniopodis} . Now it has been proved in \cite{Flocchiniopodis} that the problem $-IL$ cannot be solved in $\mathcal{FSTA}^{F}$. Hence the problem $-IL$ cannot be solved in $\mathcal{OBLOT}^{F}$ as well. Hence $-IL$ can be solved in $\mathcal{FCOM}^{A}$ but not in $\mathcal{OBLOT}^{F}$.

Again in \cite{Flocchiniopodis}, the authors gave a problem $SRO$ or Shrinking Rotation which is solvable in $\mathcal{OBLOT}^{F}$, but not solvable in $\mathcal{FCOM}^{S}$ hence in $\mathcal{FCOM}^{A}$. Hence our theorem is proved.\qed

\end{proof}

 \subsection{Independent Oscillating Problem}
 
 It has already been proved that there exists a problem which is not solvable in $\mathcal{FSTA}^{F}$, hence in $\mathcal{FSTA}^{A}$ but solvable in $\mathcal{FCOM}^{S}$, i.e., $-IL$ problem. In this section we introduce the $Independent$ $Oscillating$ $Problem$ defined in Definition \ref{IOP} which is solvable in $\mathcal{FSTA}^{A}$ but not in $\mathcal{FCOM}^{S}$. 
 
 \begin{figure}[htb!]
\centering
\includegraphics[width=10cm, height=0.5cm]{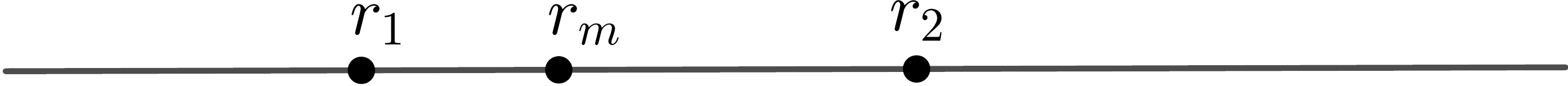}
\caption{Initial configuration of $Independent$ $Oscillating$ $Problem$}
\label{fig:OCP0}
\end{figure}
 
 \begin{definition}\label{IOP}
    \textbf{Independent Oscillating Problem (IOP):} Let three robots be placed initially on a straight line $L$ at different positions on the plane arbitrarily. Let $r_m$ be the robot that lie in between the remaining two terminal robots $r_1$ and $r_2$ (Figure \ref{fig:OCP0}). For i=1,2, let $x_i$ be the initial distance between $r_i$ and $r_m$. The problem requires that for each i=1,2, whenever $r_i$ is activated, then if it is at distance $x_i$ from $r_m$, it moves away $x_i$ distance from $r_m$ along $L$ and, if it is at distance $2x_i$ from $r_m$, it moves closer $x_i$ distance towards $r_m$ along $L$.\\
    More precisely, for i=1,2, let $d_i(t)$ denote the distance between robots $r_m$ and $r_i$ at time $t$. The Independent Oscillating Problem requires each robot $r_i$, starting from an arbitrary distance $d_i(t_0)>0$ at time $t_0$, to move so that there exists a monotonically increasing infinite sequence of time instances $t^i_0,t^i_1,t^i_2,\dots$ such that
    $d_i(t^i_{2k})=d_i(t^i_0)$ and $d_i(t^i_{2k-1})=2d_i(t^i_0)$ for all $k=1,2,3,\dots$, and $\forall h', h'' \in \left[ t^i_{2k}, t^i_{2k+1}   \right]$ and $h' < h'' $, $d_i(h') \leq d_i(h'')$ and $\forall h', h'' \in \left[ t^i_{2k - 1}, t^i_{2k}   \right]$ and $h' < h'' $, $d_i(h') \geq d_i(h'')$.
 \end{definition}
 
\begin{algorithm}[H]
$M$ = position of the robot $r_m$\;
$C$ = position of the robot $r$\;
\eIf{$r=r_m$}
     {do not do anything}
     {
        \If{r.light=NIL}
        {
            Set the r.light to $RED$\;
            Move away $CM$ distance along $L$ from the closest robot\; 
        }
        \If{r.light=RED}
        {
            Set the r.light to $NIL$\;
             Move closer $\frac{CM}2$ distance along $L$ to the closest robot\;
        }
     }   
 \caption{Algorithm $AlgoIOP$ for $IOP$ in $\mathcal{FSTA}^{A}$ executed by each robot $r$, initially r.light $\leftarrow$ $NIL$; }
\end{algorithm} 
 
\paragraph{ Description and Correctness  of the Algorithm $AlgoIOP$} 
\begin{figure}[htb!]
\begin{subfigure}{1.0\textwidth}
\centering
\includegraphics[width=0.8\linewidth]{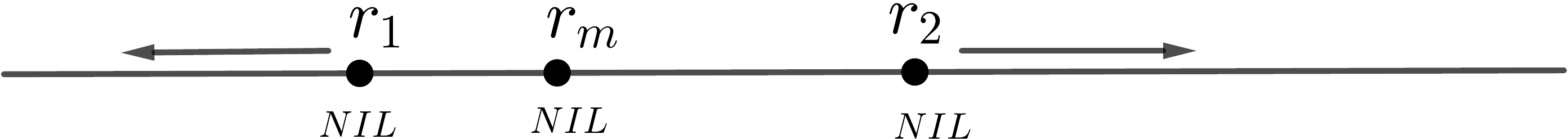}
\subcaption{When the internal light of the terminal robots are $NIL$}
\label{fig:OCP1}
\end{subfigure}
\begin{subfigure}{1.0\textwidth}
\centering
\includegraphics[width=.8\linewidth]{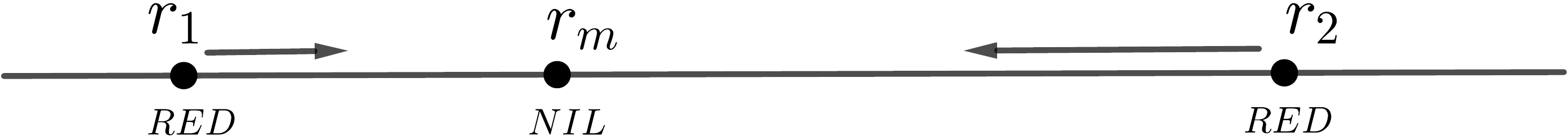}
\subcaption{When the internal light of the terminal robots are $RED$ }
\label{fig:OCP2}
\end{subfigure}
\caption{Illustration of Algorithm $AlgoIOP$}
\label{}
\end{figure}
The internal light used by the robots can take two colors: $NIL$, $RED$. Initially all robots have their internal lights set to $NIL$. A robot $r$ when activated takes a snapshot of its surrounding in the $LOOK$ phase and recognises itself either as the middle robot $r_m$ or as a terminal robot $r_i$. If it recognises itself as $r_m$ then it does nothing. Moreover, at any time instance the this robot stays in the middle according to the algorithm. So whenever $r_m$ gets activated, it neither changes its position nor its internal color. When a terminal robot gets activated for the first time then it recognises itself as a terminal robot from the snapshot taken in $LOOK$ phase. From the internal light it decides whether to move closer to the middle robot or move away from the middle robot. Let $C,M$ be the position of the current terminal robot and $r_m$ in a snapshot. If the internal light is $NIL$ then it determines that in this round it has to move away $CM$ distance from the middle robot along $L$, and if the internal light is $RED$ then it determines that in this round it has to move closer $\frac{CM}2$ distance towards the middle robot along $L$.

From the above discussion it is easy to observe that $IOP$ is solvable  by algorithm $AlgoIOP$ in $\mathcal{FSTA}^A$. We record this result in Lemma \ref{lemma5} .
\begin{lemma}\label{lemma5}
  $\forall$  $R\in \mathcal{R}_3$, $IOP\in\mathcal{FSTA}^A(R)$.
\end{lemma}

\begin{lemma}\label{lemma6}
  $\exists$  $R\in \mathcal{R}_3$, $IOP\not\in\mathcal{FCOM}^S(R)$.
\end{lemma}
\begin{proof}
If possible let there be an algorithm \texttt{A} which solves the $IOP$ in $\mathcal{FCOM}^S$. If both of the terminal robots do not move ever, then it contradicts the correctness of the algorithm \texttt{A}. Let at $k^{th}$ round $r_1$ makes the first non null movement, according to the problem that is, to move $x_1$ distance away from $r_m$ along $L$.  If the scheduler activates only $r_1$ at both the $k^{th}$ and $(k+1)^{th}$ round, then $r_1$ will not be able to differentiate between the situation in $k^{th}$ round and $(k+1)^{th}$ round and in round $k+1$, it moves $2x_1$ distance away from $r_m$ along $L$, because the distances of the robots from each other at any particular round is arbitrary and external lights of $r_2$ and $r_m$ are same in the both $k$ and $(k+1)^{th}$ round. This contradicts the correctness of \texttt{A}.\qed   
\end{proof}

\begin{lemma}\label{lemma70}
\cite{Flocchiniopodis} $\exists R\in\mathcal{R}_3, -IL\not\in\mathcal{FSTA}^F(R)$.
\end{lemma}
\begin{lemma}\label{lemma80}
\cite{Flocchiniopodis} $\forall R\in\mathcal{R}_3, -IL\in\mathcal{FCOM}^S(S)$.
\end{lemma}

Next in \cite{Flocchiniopodis} Lemma \ref{lemma70} and \ref{lemma80} are shown regarding the $-IL$ problem which prove that $-IL$ problem can not be solved in $\mathcal{FSTA}^F$, hence in $\mathcal{FSTA}^A$  but can be solved in $\mathcal{FCOM}^S$. Also in this paper we have shown that $-IL$ is solvable in $\mathcal{FCOM}^A$. Therefore from Lemmas \ref{lemma5}-\ref{lemma6} and \ref{lemma70}-\ref{lemma80}, we can conclude:

\begin{theorem}
$\mathcal{FCOM}^S\perp\mathcal{FSTA}^A.$
\end{theorem}

Also, given the fact that $IOP$ cannot be solved in $\mathcal{FCOM}^S$ implies that $IOP$ cannot be solved in $\mathcal{FCOM}^A$, we get the following result:

\begin{theorem}
$\mathcal{FCOM}^A\perp\mathcal{FSTA}^A.$
\end{theorem}

\section{Conclusion}
In this paper we have extended the study of computational relationship between models of \cite{Flocchiniopodis} under asynchronous scheduler. It has only been recently proved in \cite{Flocchiniopodis} that $\mathcal{LUMI}^{A}$ $\perp$ $\mathcal{OBLOT}^{F}$, which had been an long-standing open problem. In this paper we have further refined that result by   proving that $\mathcal{FCOM}^{A}$ $\perp$ $\mathcal{OBLOT}^{F}$, and $\mathcal{FSTA}^{A}$ $\perp$ $\mathcal{OBLOT}^{F}$,  as $\mathcal{FCOM}^{A}$ and $\mathcal{FSTA}^{A}$ are sub-models of $\mathcal{LUMI}^{A}$, with either only the power of communication or memory compared to the power of both communication and memory of $\mathcal{LUMI}^{A}$. In fact we have gone one step further and even proved that $\mathcal{FCOM}^{A}$ $\perp$ $\mathcal{FSTA}^{F}$. This is in contrast to the result $\mathcal{FCOM}^{F}$ $>$ $\mathcal{FSTA}^{A}$.  Perhaps the most significant contribution of our paper is to find out the exact computational relationship between $\mathcal{FSTA}^{A}$ and $\mathcal{FCOM}^{A}$. In this paper we  have proved that, $\mathcal{FSTA}^{A} \perp \mathcal{FCOM}^{A}$, proving that the finite memory model and the finite communication model are incomparable under asynchronous scheduler. Still there are many relations whose nature is still yet to be resolved. The relations that are yet to be resolved are: 

\begin{enumerate}
    \item $\mathcal{OBLOT}^{S} \ Vs.\  \mathcal{OBLOT}^{A}$
    \item $\mathcal{FSTA}^{S} \ Vs.\  \mathcal{FSTA}^{A}$
    \item $\mathcal{FSTA}^{A} \ Vs.\  \mathcal{OBLOT}^{S}$
    \item $\mathcal{FCOM}^{S} \ Vs.\  \mathcal{FCOM}^{A}$
    \item $\mathcal{FCOM}^{A} \ Vs.\  \mathcal{OBLOT}^{S}$
\end{enumerate}

The first of these questions is answered as $\mathcal{OBLOT}^{S} > \mathcal{OBLOT}^{A}$ in graph domain \cite{NavarraInfComp2018}. But in the background of continuous setting, where the robots can detect multiplicity, the question is yet to be resolved.
The rest of the questions remain open for future investigations.
Also we assumed in our setting rigid movements and the presence of chirality. It would be interesting to explore the nature of the relations by changing one or both the assumptions. In the papers \cite{LEnhancedRobotGNavarra} and  \cite{NavarraInfComp2018} characterization of some of the relations have been done in graph environment or discrete setting. Also, it would be interesting to characterize the remaining relations of the model in discrete setting. 

\paragraph{Acknowledgement}
Partial work of this paper was communicated to the conference $CALDAM,2021$, so the authors would like to thank the anonymous reviewers of that conference for their valuable comments which have helped us to improve this work.

\bibliographystyle{plain}
\bibliography{finalcomp}

\end{document}

%% file: oc.pdf_tex
\begingroup%
  \makeatletter%
  \providecommand\color[2][]{%
    \errmessage{(Inkscape) Color is used for the text in Inkscape, but the package 'color.sty' is not loaded}%
    \renewcommand\color[2][]{}%
  }%
  \providecommand\transparent[1]{%
    \errmessage{(Inkscape) Transparency is used (non-zero) for the text in Inkscape, but the package 'transparent.sty' is not loaded}%
    \renewcommand\transparent[1]{}%
  }%
  \providecommand\rotatebox[2]{#2}%
  \newcommand*\fsize{\dimexpr\f@size pt\relax}%
  \newcommand*\lineheight[1]{\fontsize{\fsize}{#1\fsize}\selectfont}%
  \ifx\svgwidth\undefined%
    \setlength{\unitlength}{622.5bp}%
    \ifx\svgscale\undefined%
      \relax%
    \else%
      \setlength{\unitlength}{\unitlength * \real{\svgscale}}%
    \fi%
  \else%
    \setlength{\unitlength}{\svgwidth}%
  \fi%
  \global\let\svgwidth\undefined%
  \global\let\svgscale\undefined%
  \makeatother%
  \begin{picture}(1,0.48192771)%
    \lineheight{1}%
    \setlength\tabcolsep{0pt}%
    \put(0.01685151,0.44115922){\color[rgb]{0,0,0}\makebox(0,0)[lt]{\lineheight{0}\smash{\begin{tabular}[t]{l}$A(x)$\end{tabular}}}}%
    \put(0,0){\includegraphics[width=\unitlength,page=1]{oc.pdf}}%
    \put(0.03889273,0.11374324){\makebox(0,0)[lt]{\lineheight{1.25}\smash{\begin{tabular}[t]{l}(I)\end{tabular}}}}%
    \put(0.40448432,0.1215218){\makebox(0,0)[lt]{\lineheight{1.25}\smash{\begin{tabular}[t]{l}(II)\end{tabular}}}}%
    \put(0.7234047,0.1215218){\makebox(0,0)[lt]{\lineheight{1.25}\smash{\begin{tabular}[t]{l}(III)\end{tabular}}}}%
    \put(0.19446363,0.4274779){\makebox(0,0)[lt]{\lineheight{1.25}\smash{\begin{tabular}[t]{l}$B(y)$\end{tabular}}}}%
    \put(0.13482811,0.28227839){\makebox(0,0)[lt]{\lineheight{1.25}\smash{\begin{tabular}[t]{l}$C(t)$\end{tabular}}}}%
    \put(0.02333563,0.17337874){\makebox(0,0)[lt]{\lineheight{1.25}\smash{\begin{tabular}[t]{l}$D(z)$\end{tabular}}}}%
    \put(0.3215132,0.43007074){\makebox(0,0)[lt]{\lineheight{1.25}\smash{\begin{tabular}[t]{l}$A$\end{tabular}}}}%
    \put(0.53412674,0.44044214){\makebox(0,0)[lt]{\lineheight{1.25}\smash{\begin{tabular}[t]{l}$B$\end{tabular}}}}%
    \put(0.31114178,0.18115731){\makebox(0,0)[lt]{\lineheight{1.25}\smash{\begin{tabular}[t]{l}$D$\end{tabular}}}}%
    \put(0.55486954,0.18375016){\makebox(0,0)[lt]{\lineheight{1.25}\smash{\begin{tabular}[t]{l}$C'$\end{tabular}}}}%
    \put(0.65080493,0.43266359){\makebox(0,0)[lt]{\lineheight{1.25}\smash{\begin{tabular}[t]{l}$A$\end{tabular}}}}%
    \put(0.87377062,0.40885252){\makebox(0,0)[lt]{\lineheight{1.25}\smash{\begin{tabular}[t]{l}$B$\end{tabular}}}}%
    \put(0.66376916,0.1759716){\makebox(0,0)[lt]{\lineheight{1.25}\smash{\begin{tabular}[t]{l}$D$\end{tabular}}}}%
    \put(0.43300568,0.43266359){\makebox(0,0)[lt]{\lineheight{1.25}\smash{\begin{tabular}[t]{l}$a$\end{tabular}}}}%
    \put(0.35262738,0.31339257){\makebox(0,0)[lt]{\lineheight{1.25}\smash{\begin{tabular}[t]{l}$a$\end{tabular}}}}%
    \put(0.50301256,0.27968555){\makebox(0,0)[lt]{\lineheight{1.25}\smash{\begin{tabular}[t]{l}$a$\end{tabular}}}}%
    \put(0.42263426,0.18893585){\makebox(0,0)[lt]{\lineheight{1.25}\smash{\begin{tabular}[t]{l}$a$\end{tabular}}}}%
    \put(0.12186387,0.4222922){\makebox(0,0)[lt]{\lineheight{1.25}\smash{\begin{tabular}[t]{l}$a$\end{tabular}}}}%
    \put(0.03629988,0.31079973){\makebox(0,0)[lt]{\lineheight{1.25}\smash{\begin{tabular}[t]{l}$a$\end{tabular}}}}%
    \put(0.62228357,0.29264979){\makebox(0,0)[lt]{\lineheight{1.25}\smash{\begin{tabular}[t]{l}$a$\end{tabular}}}}%
    \put(0.74414746,0.4274779){\makebox(0,0)[lt]{\lineheight{1.25}\smash{\begin{tabular}[t]{l}$a$\end{tabular}}}}%
    \put(0.09074969,0.36524953){\makebox(0,0)[lt]{\lineheight{1.25}\smash{\begin{tabular}[t]{l}$\frac{a}{\sqrt[2]{2}}$\end{tabular}}}}%
    \put(0.06222836,0.23301427){\makebox(0,0)[lt]{\lineheight{1.25}\smash{\begin{tabular}[t]{l}$\frac{a}{\sqrt[2]{2}}$\end{tabular}}}}%
    \put(0.15816375,0.33154251){\makebox(0,0)[lt]{\lineheight{1.25}\smash{\begin{tabular}[t]{l}$\frac{a}{\sqrt[2]{2}}$\end{tabular}}}}%
    \put(0.82193295,0.27449984){\makebox(0,0)[lt]{\lineheight{1.25}\smash{\begin{tabular}[t]{l}$\frac{3a}{\sqrt[2]{2}}$\end{tabular}}}}%
    \put(0.91527548,0.07225768){\makebox(0,0)[lt]{\lineheight{1.25}\smash{\begin{tabular}[t]{l}$C''$\end{tabular}}}}%
    \put(0.24178311,0.16399764){\color[rgb]{0,0,0}\makebox(0,0)[lt]{\lineheight{0}\smash{\begin{tabular}[t]{l} \end{tabular}}}}%
  \end{picture}%
\endgroup%

%% file: ilnew.pdf_tex
\begingroup%
  \makeatletter%
  \providecommand\color[2][]{%
    \errmessage{(Inkscape) Color is used for the text in Inkscape, but the package 'color.sty' is not loaded}%
    \renewcommand\color[2][]{}%
  }%
  \providecommand\transparent[1]{%
    \errmessage{(Inkscape) Transparency is used (non-zero) for the text in Inkscape, but the package 'transparent.sty' is not loaded}%
    \renewcommand\transparent[1]{}%
  }%
  \providecommand\rotatebox[2]{#2}%
  \newcommand*\fsize{\dimexpr\f@size pt\relax}%
  \newcommand*\lineheight[1]{\fontsize{\fsize}{#1\fsize}\selectfont}%
  \ifx\svgwidth\undefined%
    \setlength{\unitlength}{666.75bp}%
    \ifx\svgscale\undefined%
      \relax%
    \else%
      \setlength{\unitlength}{\unitlength * \real{\svgscale}}%
    \fi%
  \else%
    \setlength{\unitlength}{\svgwidth}%
  \fi%
  \global\let\svgwidth\undefined%
  \global\let\svgscale\undefined%
  \makeatother%
  \begin{picture}(1,0.49493813)%
    \lineheight{1}%
    \setlength\tabcolsep{0pt}%
    \put(0,0){\includegraphics[width=\unitlength,page=1]{ilnew.pdf}}%
    \put(0.01312897,0.24156817){\color[rgb]{0,0,0}\makebox(0,0)[lt]{\lineheight{0}\smash{\begin{tabular}[t]{l}$a$\end{tabular}}}}%
    \put(0.20576541,0.24801964){\makebox(0,0)[lt]{\lineheight{1.25}\smash{\begin{tabular}[t]{l}$b$\end{tabular}}}}%
    \put(0.15735002,0.04225421){\makebox(0,0)[lt]{\lineheight{1.25}\smash{\begin{tabular}[t]{l}$c$\end{tabular}}}}%
    \put(0.41197876,0.40515264){\makebox(0,0)[lt]{\lineheight{1.25}\smash{\begin{tabular}[t]{l}$a$\end{tabular}}}}%
    \put(0.41197876,0.21915744){\makebox(0,0)[lt]{\lineheight{1.25}\smash{\begin{tabular}[t]{l}$b$\end{tabular}}}}%
    \put(0.41197876,0.03316227){\makebox(0,0)[lt]{\lineheight{1.25}\smash{\begin{tabular}[t]{l}$c$\end{tabular}}}}%
    \put(0.35585312,0.43683967){\makebox(0,0)[lt]{\lineheight{1.25}\smash{\begin{tabular}[t]{l}$P_1$\end{tabular}}}}%
    \put(0.60842336,0.40741779){\makebox(0,0)[lt]{\lineheight{1.25}\smash{\begin{tabular}[t]{l}$a$\end{tabular}}}}%
    \put(0.6221378,0.06646193){\makebox(0,0)[lt]{\lineheight{1.25}\smash{\begin{tabular}[t]{l}$c$\end{tabular}}}}%
    \put(0,0){\includegraphics[width=\unitlength,page=2]{ilnew.pdf}}%
    \put(0.06294001,0.04225421){\makebox(0,0)[lt]{\lineheight{1.25}\smash{\begin{tabular}[t]{l}(I)\end{tabular}}}}%
    \put(0.31470005,0.03983347){\makebox(0,0)[lt]{\lineheight{1.25}\smash{\begin{tabular}[t]{l}(II)\end{tabular}}}}%
    \put(0.75043857,0.00594269){\makebox(0,0)[lt]{\lineheight{1.25}\smash{\begin{tabular}[t]{l}(III)\end{tabular}}}}%
    \put(0,0){\includegraphics[width=\unitlength,page=3]{ilnew.pdf}}%
    \put(0.77464626,0.06888271){\makebox(0,0)[lt]{\lineheight{1.25}\smash{\begin{tabular}[t]{l}$b$\end{tabular}}}}%
    \put(0.83516546,0.03499192){\makebox(0,0)[lt]{\lineheight{1.25}\smash{\begin{tabular}[t]{l}$P_2$\end{tabular}}}}%
    \put(0.22573676,0.15311366){\color[rgb]{0,0,0}\makebox(0,0)[lt]{\lineheight{0}\smash{\begin{tabular}[t]{l} \end{tabular}}}}%
  \end{picture}%
\endgroup%

%% file: finalcomp.bbl
\begin{thebibliography}{10}

\bibitem{Daslumi16}
Shantanu Das, Paola Flocchini, Giuseppe Prencipe, Nicola Santoro, and Masafumi
  Yamashita.
\newblock Autonomous mobile robots with lights.
\newblock {\em Theor. Comput. Sci.}, 609:171--184, 2016.

\bibitem{LEnhancedRobotGNavarra}
Mattia D'Emidio, Daniele Frigioni, and Alfredo Navarra.
\newblock Synchronous robots vs asynchronous lights-enhanced robots on graphs.
\newblock {\em Electron. Notes Theor. Comput. Sci.}, 322:169--180, 2016.

\bibitem{NavarraInfComp2018}
Mattia D'Emidio, Gabriele~Di Stefano, Daniele Frigioni, and Alfredo Navarra.
\newblock Characterizing the computational power of mobile robots on graphs and
  implications for the euclidean plane.
\newblock {\em Inf. Comput.}, 263:57--74, 2018.

\bibitem{FlocchiniPSW05}
Paola Flocchini, Giuseppe Prencipe, Nicola Santoro, and Peter Widmayer.
\newblock Gathering of asynchronous robots with limited visibility.
\newblock {\em Theor. Comput. Sci.}, 337(1-3):147--168, 2005.

\bibitem{Conrend}
Paola Flocchini, Nicola Santoro, Giovanni Viglietta, and Masafumi Yamashita.
\newblock Rendezvous with constant memory.
\newblock {\em Theor. Comput. Sci.}, 621:57--72, 2016.

\bibitem{Flocchiniopodis}
Paola Flocchini, Nicola Santoro, and Koichi Wada.
\newblock On memory, communication, and synchronous schedulers when moving and
  computing.
\newblock In {\em 23rd International Conference on Principles of Distributed
  Systems, {OPODIS} 2019, December 17-19, 2019, Neuch{\^{a}}tel, Switzerland},
  pages 25:1--25:17. Schloss Dagstuhl - Leibniz-Zentrum f{\"{u}}r Informatik,
  2019.

\bibitem{IzumiSKIDWY12}
Taisuke Izumi, Samia Souissi, Yoshiaki Katayama, Nobuhiro Inuzuka, Xavier
  D{\'{e}}fago, Koichi Wada, and Masafumi Yamashita.
\newblock The gathering problem for two oblivious robots with unreliable
  compasses.
\newblock {\em {SIAM} J. Comput.}, 41(1):26--46, 2012.

\bibitem{peterscomparison}
Tom Peters.
\newblock Comparison of scheduler models for distributed systems of luminous
  robots.

\bibitem{SuzukiY99}
Ichiro Suzuki and Masafumi Yamashita.
\newblock Distributed anonymous mobile robots: Formation of geometric patterns.
\newblock {\em {SIAM} J. Comput.}, 28(4):1347--1363, 1999.

\bibitem{YamashitaS10}
Masafumi Yamashita and Ichiro Suzuki.
\newblock Characterizing geometric patterns formable by oblivious anonymous
  mobile robots.
\newblock {\em Theor. Comput. Sci.}, 411(26-28):2433--2453, 2010.

\bibitem{YamauchiUKY17}
Yukiko Yamauchi, Taichi Uehara, Shuji Kijima, and Masafumi Yamashita.
\newblock Plane formation by synchronous mobile robots in the three-dimensional
  euclidean space.
\newblock {\em J. {ACM}}, 64(3):16:1--16:43, 2017.

\end{thebibliography}
